\documentclass[preliminary, copyright, creativecommons]{eptcs}
\pdfoutput=1
\providecommand{\event}{\href{http://discotec2015.inria.fr/workshops/ice-2015/}{ICE 2015}}
\usepackage[utf8]{inputenc}
\usepackage[T1]{fontenc}
\usepackage[english]{babel}
\usepackage{mathtools}
\usepackage{amssymb}
\usepackage{amsthm}
\usepackage{amsmath}
\usepackage{doi}
\usepackage{tabularx}
\usepackage{multirow}
\usepackage{ebproof}
\usepackage{tikz}
\usetikzlibrary{arrows, calc, positioning, shapes}
\usetikzlibrary{decorations.text} 
\usepackage{wrapfig}
\usepackage{xspace}
\usepackage[autostyle = true, english = american, strict = true, autopunct=true]{csquotes}
\usepackage[final, babel]{microtype}
\usepackage{xcolor}
\definecolor{colorp1}{RGB}{229,245,249}
\definecolor{colorp2}{RGB}{153,216,201}
\definecolor{colorp3}{RGB}{44,162,95}

\AfterPreamble{
	\hypersetup{
	 plainpages=false,
	 bookmarksnumbered=true,
	 bookmarksopen=true,
	 bookmarksdepth=2,
	 breaklinks=true,
	 linkbordercolor=colorp1,
	 urlbordercolor=colorp2,
	 citebordercolor=colorp2,
	 pdfinfo={
		 Author={Clément Aubert, Ioana Cristescu},
		 Title={Reversible Barbed Congruence on Configuration Structures},
		 Keywords={Formal semantics, Process algebras and calculi, Reversible CCS, Hereditary history preserving bisimulation, Strong barbed congruence, Contextual caracterisation},
		 Creator={PDFLaTeX},
		 Producer={PDFLaTeX},
		 Lang={en}
		 }
	}
}

\theoremstyle{plain}
\newtheorem{theorem}{Theorem}
\newtheorem{lemma}{Lemma}
\newtheorem{proposition}{Proposition}

\theoremstyle{definition}
\newtheorem{definition}{Definition}
\newtheorem{remark}{Remark}
\newtheorem{notations}{Notations}
\newtheorem{example}{Example}[section]

\newcommand{\fw}{\mathrel{\longrightarrow}}
\newcommand{\bw}{\mathrel{\rightsquigarrow}}
\newcommand{\fbw}{\mathrel{\twoheadrightarrow}} 
\newcommand{\red}{\mathrel{\fw}}
\newcommand{\revred}{\mathrel{\bw}}
\newcommand{\redl}[1]{\mathrel{\overset{#1}{\red}}}
\newcommand{\revredl}[1]{\mathrel{\overset{#1}{\revred}}}
\newcommand{\fbl}[1]{\mathrel{\overset{#1}{\fbw}}}

\newcommand{\idmem}[2]{#1 : #2}
\newcommand{\compfleche}[3]{\overset{\idmem{#1}{#2}}#3}
\newcommand{\fwlts}[2]{\compfleche{#1}{#2}{\fw}}
\newcommand{\bwlts}[2]{\compfleche{#1}{#2}{\bw}}
\newcommand{\fbwlts}[2]{\compfleche{#1}{#2}{\fbw}}

\newcommand{\congru}{\equiv}
\newcommand{\sbisim}{\sim}
\newcommand{\scbisim}{\mathrel{\sbisim^{\tau}}}
\newcommand{\sbbisim}{\mathrel{\dot{\sbisim}^{\tau}}}
\newcommand{\sbfbb}{\mathrel{\dot{\sbisim^{\tau}}}}
\newcommand{\sbfbc}{\mathrel{\sbisim^{\tau}}}
\newcommand{\Forw}[1]{F_{#1}}
\newcommand{\Backw}[1]{B_{#1}} 
\newcommand{\emptymem}{\emptyset}

\newcommand{\orig}[1]{O_{#1}}
\newcommand{\fork}{\curlyvee}

\DeclareMathOperator{\ids}{\mathsf{I}}

\DeclareMathOperator{\address}{ad}
\DeclareMathOperator{\collapse}{collapse}
\DeclareMathOperator{\card}{Card}

\newcommand{\funaddress}[3]{\address_{#1}(#2, #3)}
\newcommand{\encc}[1]{\ensuremath{[\![#1]\!]}}
\newcommand{\enc}[1]{\encc{#1}}
\newcommand{\encr}[1]{\ensuremath{[\![#1]\!]}}

\newcommand{\names}{\ensuremath{\mathsf{N}}}
\newcommand{\labels}{\ensuremath{\mathsf{L}}}

\newcommand{\set}[1]{\ensuremath{\{#1\}}}
\newcommand{\out}[1]{\bar{#1}}
\newcommand{\power}{\ensuremath{\mathcal{P}}}
\newcommand{\conf}{\ensuremath{\mathcal{C}}}
\newcommand{\rel}{\ensuremath{\mathcal{R}}}
\newcommand{\srel}{\ensuremath{\mathcal{R}}}
\newcommand{\mem}[1]{\ensuremath{\langle#1\rangle}}
\newcommand{\labl}{\ell}
\newcommand{\restr}{\mathord{\upharpoonleft}}
\newcommand{\erase}{\ensuremath{\varepsilon}}

\newcommand{\st}{s.t.\ }
\newcommand{\withoutlog}{w.l.o.g.\xspace}
\newcommand{\BNFsepa}{\enspace \Arrowvert \enspace}

\title{Reversible Barbed Congruence on Configuration Structures
\thanks{This work was partly supported by the ANR-14-CE25-0005 \href{http://lipn.univ-paris13.fr/~mazza/Elica/}{ELICA} and the ANR-11-INSE-0007 \href{http://www.pps.univ-paris-diderot.fr/~jkrivine/ANR/REVER/ANR_REVER/Welcome.html}{REVER}.}
}
\def\titlerunning{Reversible Barbed Congruence on Configuration Structures}

\author{
Clément Aubert
\institute{INRIA}
\institute{Université Paris-Est, LACL (EA 4219), UPEC, F-94010 Créteil, France}
\email{clement.aubert@lacl.fr}
\and
Ioana Cristescu
\institute{Univ. Paris Diderot, Sorbonne Paris Cité, P.P.S., UMR 7126, F-75205 Paris, France}
\email{ioana.cristescu@pps.univ-paris-diderot.fr}
}
\def\authorrunning{C. Aubert \& I. Cristescu}

\begin{document}
\maketitle

\begin{abstract}
	A standard contextual equivalence for process algebras is strong barbed congruence.
	Configuration structures are a denotational semantics for processes in which one can define equivalences that are more discriminating, i.e.\ that distinguish the denotation of terms equated by barbed congruence.
	Hereditary history preserving bisimulation (HHPB) is such a relation. 	We define a strong back and forth barbed congruence using a reversible process algebra and show that the relation induced by the back and forth congruence is equivalent to HHPB, providing a contextual characterization of HHPB.
\end{abstract}

\section*{Introduction}
\addcontentsline{toc}{section}{Introduction}
A standard notion of equivalence for process algebras identifies processes that interact the same way with the environment.
Reduction congruence~\cite{Milner1992} is a standard relation that equates terms capable of simulating each other's reductions in any context.
However observing only the reductions is a too coarse relation.
A predicate, called a \emph{barb}, is then defined to handle an extra observation on processes: the channel on which they communicate with the environment.

Configuration structures---also called \emph{stable families}~\cite{Winskel1982} or \emph{stable configuration structures}~\cite{Glabbeek2001}---are an extensional representation of processes, which explicit all possible future behaviours.
It consists of a family of sets, where each set is called a \emph{configuration} and stands for a reachable state in the run of the process. The elements of the sets, called \emph{events}, represent the actions the process triggered so far.
The inclusion relation between configurations stands for the possible paths followed by the execution.
The encoding of terms of the Calculus of Communicating Systems (CCS)---a simple process algebra---in configuration structures~\cite{Winskel1982,Glabbeek2001} settled configuration structures as a denotational model for concurrency.

Configuration structures are \enquote{true concurrency} models, as opposed to process algebras, which use an interleaving representation of concurrency.
It is hard to deduce in an interleaving semantics the relationships between events, such as whether two events are independent or not, whereas they are explicit or easily inferred in a truly concurrent semantics.

On such structures, the equivalence relations defined are more discriminating: it is possible to move \enquote{up and down} in the lattice, whereas in the operational semantics, only forward transitions have to be simulated.
As an example, consider the processes $a.0|b.0$ and $a.b.0+b.a.0$ that are bisimilar in CCS but whose causal relations between events differ.
In particular we investigate hereditary history preserving bisimulation (HHPB), which equates structures that can simulate each others' forward and backward moves.
It is the canonical equivalence on configuration structures as it respects the causality and concurrency relations between events and admits a categorical representation~\cite{Joyal1996b}.

Reversibility allows to define HHPB in an operational setting, by simply adding to processes the capability to undo previous computations.
A term can then either continue its forward execution or backtrack up to a point in the past and resume from there.
Reversible process algebras are interesting in their own right~\cite{Lanese2010,Cristescu2013}, but we focus in this paper on their capability to simulate the back-and-forth behaviour of configuration structures.
To ensure that the backward reduction of CCS indeed corresponds to the backward moves of its denotational representation, one has to prove that the \emph{labelled transition system} is \emph{prime}~\cite{Nielsen1981}.
It was already done for CCSK~\cite{Phillips2007}, a reversible calculus that has a causal consistent backtracking machinery.
In this paper we use RCCS~\cite{Danos2004} a causal consistent reversible variant of CCS whose syntax is given in \autoref{rccs-syntax}.

In reversible calculi one is also interested in a contextual equivalence for processes. Traditional equivalences, defined only on forward transitions, are inappropriate for processes that can do back-and-forth reductions.
Strong back-and-forth bisimulation~\cite{Lanese2010} is more adapted but it is not contextual.
Hence we introduce the barbed back-and-forth congruence on RCCS terms (\autoref{contextual-equiv-rccs}) which corresponds to the barbed congruence of CCS except that backwards reductions are also observed.

Configuration structures (\autoref{conf-struct}) lacks a notion of contextual equivalence, because the context is a notion specific to the operational semantics.
Hence it makes sense to consider context only for configuration structures that represent an operational term (\autoref{RCCS-in-conf}).
We introduce in \autoref{def-bisim-cs} the correct notions and relations on those structures.
The contextual equivalence on processes induces a relation on the denotation of these processes and this relation corresponds to HHPB
(\autoref{sec:res}).

Similarly to the proof in CCS, the correspondence between a contextual equivalence and a non contextual one necessits to approximate hhbp with (a family of) inductive relations defined on configuration structures.
If we are interested only in the forward direction (as in CCS), the inductive reasoning starts with the empty set, and constructs the bisimilarity relation by adding pairs of configuration reachable in the same manner from the empty set. However, to approximate hhbp, we need to have an inductive reasoning on the backward transition as well (\autoref{ForwBackwDef}).
These relations are of major importance to prove our main theorem (\autoref{main-thm}), as they re-introduce the possibility of an inductive reasoning thanks to a stratification of the HHPB relation.

Hhpb is equivalent to strong bisimulation on reversible CCS~\cite{Phillips2012}, thus it can be characterised as a non contextual equivalence on processes. One can then prove the main result of the paper by showing that in RCCS strong bisimulation and strong barbed congruence equate the same terms. We chose to use configuration structures instead, as we plan to investigate weak equivalences on reversible process algebra and their correspondence in denotational semantics.

Our work is restrained to processes that forbid \emph{some sort of auto-concurrency} (see \autoref{rem-concur}) and that are \enquote{collapsed} (\autoref{def:collapse}): we need to uniquely identify open configurations using only the label and the order of the events.
The \enquote{equidepth auto-concurrency}~\cite{Phillips2012} does not help.
\section{RCCS syntax and bisimulation}
\label{sec:rccs}
RCCS is a reversible variant of CCS, that allows computations to backtrack, hence introducing the notions of \emph{forward} and \emph{backward} transitions.
A mechanism of memories attached to processes store the relevant information to eventually do backward steps.

In a sequential setting backtracking follows the exact order of the forward computation. This is too strict for a concurrent calculus where independent processes can fire independent actions. The order of these actions in the forward direction is just \emph{temporal} and \emph{not causal}, and thus it should be allowed to backtrack them in any order.
On the other hand, too much liberty in backtracking could allow the system to access states that were not reachable with forward transitions alone.

\subsection{RCCS syntax}
\label{rccs-syntax}
\begin{notations}
	Let $\names=\set{a,b,\dots}$ a set of \emph{names}, \(\ids=\set{i,j,\dots}\) a set of \emph{identifiers}.
	An \emph{action} is an input (resp.\ output) on a channel $a$, labelled $a$ (resp.\ $\out{a}$), or a synchronisation with the label \( (a,\out a)\), sometimes denoted $\tau$.
	Each action \(a\) has a dual written \(\out{a}\), we let \(\out{\out{a}} = a\) and \(\out{\tau} = \tau\). Denote \(\labels=\set{\alpha,\beta,..}\) the set of \emph{labels}.
\end{notations}
CCS processes are build using prefix, sum, parallel composition and restriction.
RCCS processes, also called \emph{monitored processes}, are built upon CCS processes by adding a memory \(m\) that acts as a stack of the previous computations.
Each entry in the memory is called an \emph{event} and has a unique identifier.
The usual~\cite{Danos2005} RCCS processes grammar is recalled in \autoref{rccs-gram}.
A memory \(\mem{i, \alpha, P}\) contains an \enquote{identifier} \(i\) that \enquote{tags} transitions: it is especially useful in the case of synchronisation (both forward and backward), for it identifies which two processes interact.
The label \(\alpha\) marks which action has been fired (in the case of a forward transition), or what action should be restored (in the case of a backward move).
Finally, \(P\) saves the whole process that has been erased when firing a sum.
The \enquote{fork symbol} \(\fork\) marks that the memory of a parallel composition has been split down to the two parts of the parallel composition.
It was handled with \(\mem{1}\) and \(\mem{2}\) (Left- and Right-fork) in previous work \cite[p.~295]{Danos2004}.

\begin{figure}
	\begin{align}
		\gamma & := a \BNFsepa \out{a}\BNFsepa\hdots & \alpha,\beta & := \gamma \BNFsepa \tau \tag{Actions} \\
		m &:= \emptymem \BNFsepa \fork . m \BNFsepa \mem{i, \alpha, P} . m \tag{Memories}\\
		P, Q & := 0 \BNFsepa \alpha.P \BNFsepa \alpha.P+\beta.Q \BNFsepa P \vert Q\BNFsepa (a)P \tag{CCS processes}\\
		R, S & := m \vartriangleright P \BNFsepa R \vert R \BNFsepa (a)R \tag{RCCS processes}
	\end{align}
	\caption{RCCS processes grammar}
	\label{rccs-gram}
\end{figure}

We can easily retrieve a CCS process from an RCCS one by erasing the memories:
\[
	\erase(m\vartriangleright P)=P\quad\erase(R\vert S)=\erase(R)\vert\erase(S)\quad\erase((a)R)=(a)\erase(R) \quad\erase(R + S)=\erase(R) + \erase(S)\]
Structural congruence on monitored processes is the smallest equivalence relation up to uniform renaming of identifiers generated by the following rules:

\begin{center}
	\begin{tabular}{c c}
		\multirow{2}*{
		\begin{prooftree}
		\Hypo{P \congru Q}
		\Infer1{\emptymem\rhd P \congru\emptymem\rhd Q}
		\end{prooftree}
		} &                                                                                     
		\(m \vartriangleright (P \vert Q) \congru (\fork . m \vartriangleright P \vert \fork . m \vartriangleright Q)\)\\
		  & \(m \vartriangleright (a)P \congru (a)m \vartriangleright P\text{ with }a\notin m\) 
	\end{tabular}
\end{center}
The left rule implies that all equivalence for CCS processes holds for RCCS processes with an empty memory.
The right rules respectively distributes the memory between two forking processes (top) and moves the restrictions at the process level (bottom).

The labelled transition system (LTS) for RCCS is given by the rules of \autoref{ltsrules}.
In the transitions \(\redl{i:\alpha}\) (resp.\ \(\revredl{i:\alpha}\)) for the \emph{forward} (resp.\ \emph{backward}) action, we have that \(i \in I\) is the event identifier, \(\ids(m)\) (resp.\ \(\ids(S)\)) is the set of identifiers occurring in \(m\) (resp.\ in \(S\)).
We use \(\fbwlts{i}{\alpha}\) as a wildcard for \(\redl{i:\alpha}\) or \(\revredl{i:\alpha}\), and if there are indices \(i_1, \hdots, i_n\) and labels \(\alpha_1, \hdots, \alpha_n\) such that \(R_1 \fbwlts{i_1}{\alpha_1} \hdots \fbwlts{i_n}{\alpha_n} R_n\), then we write \(R_1 \fbw^{\star} R_n\). We sometimes omit the identifier or the label in the transition.
The trace is unique up to renaming of the indices.

When a prefix is consumed we add in the memory an event consisting of an unique identifier, the label consumed and the discarded part of the non-deterministic sum. Then backtracking removes an event at the top of a memory and restores the prefix and the non-deterministic sum.
Synchronization, forward or backward (syn), requires the two synchronization partners to agree on the event identifier and trigger the transitions simultaneously.
The requirement that \(i \notin \ids(S)\) for the parallel composition (par.) ensures the uniqueness of the event identifiers in the forward direction and prevents a part of a previous synchronization to backtrack alone in the backward direction.

\begin{figure}
	{\centering
		\begin{tabular}{c c c c}
			\begin{prooftree}
			\Hypo{R \fbwlts{i}{\gamma} R'}
			\Hypo{S \fbwlts{i}{\out{\gamma}} S'}
			\Infer2[syn.]{R \vert S \fbwlts{i}{\tau} R' \vert S'}
			\end{prooftree}
			  &   
			\begin{prooftree}
			\Hypo{R \fbwlts{i}{\alpha} R'}
			\Infer1[par.]{R \vert S \fbwlts{i}{\alpha} R' \vert S}
			\end{prooftree}
			  &   
			\begin{prooftree}
			\Hypo{R \fbwlts{i}{\alpha} R'}
			\Hypo{a \notin \alpha}
			\Infer2[res.]{(a) R \fbwlts{i}{\alpha} (a) R'}
			\end{prooftree}
			  &   
			\begin{prooftree}
			\Hypo{R_1 \congru R \fbwlts{i}{\alpha} R' \congru R_1'}
			\Infer1[$\congru$]{R_1 \fbwlts{i}{\alpha} R_1'}
			\end{prooftree}
			\\[2em]
			\multicolumn{2}{c}{
			\begin{prooftree}
			\Hypo{}
			\Infer1[act.]{m \vartriangleright \alpha . P + Q \fwlts{i}{\alpha} \mem{i, \alpha, Q} . m \vartriangleright P}
			\end{prooftree}
			}
			  &   
			\multicolumn{2}{c}{
			\begin{prooftree}
			\Hypo{}
			\Infer1[act.$_*$ ]{\mem{i, \alpha, Q} . m \vartriangleright P \bwlts{i}{\alpha} m \vartriangleright \alpha . P + Q }
			\end{prooftree}
			}
		\end{tabular}
		\\[2em]
		The rule act. and act$_*$ apply iff \(i \notin \ids(m)\), the rule par. applies iff \(i \notin \ids(S)\).
	}
	\caption{Rules of the LTS}\label{ltsrules}
\end{figure}

\begin{example}
	The process \(\fork . \mem{i, \alpha, \alpha'.0} . \emptymem \triangleright P \mid \mem{j, \beta, \beta'.0}. \emptymem \triangleright Q\) highlights that not all syntactically correct processes have an operational meaning.
	This term cannot be obtained by a forward computation from a CCS process, somehow \enquote{its memory is broken}.
	Without \(\fork\), one could backtrack to \(\emptyset \vartriangleright \alpha . P + \alpha' .0 \mid \emptyset \vartriangleright \beta . Q + \beta' . 0\), but this terms violate the structural congruence.
\end{example}

The semantically correct processes are called \emph{coherent} and are defined as follows:

\begin{definition}[Coherent process and \(\orig{R}\)]
	A RCCS process \(R\) is \emph{coherent} if there exists a CCS process \(P\) such that \(\emptymem \vartriangleright P \fw^{\star} R\).
	This process \(P\) is unique up to structural congruence and we write it \(\orig{R}\).
\end{definition}

Backtracking is not deterministic, but it is noetherian and confluent~\cite[Lemma~1]{Danos2005}, hence the uniqueness.
Actually, coherence of processes comes from the coherence relation defined on memories~\cite[Definition~1]{Danos2004} and implies that in a coherent term, memories are unique.
Moreover, coherence is preserved by transitions and structural congruence.

\subsection{A contextual equivalence for RCCS}
\label{contextual-equiv-rccs}
Let us now revisit the barbed congruence of CCS~\cite{Milner1992} in the case of RCCS.
For that we need the right notions of context and barb in the reversible setting.

Choosing the right notion of context is subtle.
A context has to become an executable process regardless of the process instantiated with it. We can distinguish three types of contexts: with an empty memory, with a non empty but coherent memory (i.e. the context can backtrack up to an empty memory regardless of the process instantiated with) or with a non coherent memory.
The later is left as future work, while the first two are equivalent: we will only, \withoutlog, consider contexts without memory.

\begin{definition}[CCS Context]
	\label{context-def}
	A context is a process with a hole:
	$C:= [~] \BNFsepa \alpha.C \BNFsepa C+P \BNFsepa C \vert P \BNFsepa (a)C$
\end{definition}
We can only instantiate a context with an RCCS process $R$ if the process has an empty memory, i.e. $R=\emptyset\rhd P$. We use the notation $C[\emptyset\rhd P]$ to denote the process $\emptyset\rhd C[P]$.
\begin{definition}[Strong commitment (barb)]
	\label{def:barb}
	We write \(R \downarrow_\alpha\) if there exists \(i \in I\) and \(R'\) such that \(R \fwlts{i}{\alpha} R'\).
\end{definition}

\begin{definition}	\label{def:sbfc_rccs}
	A \emph{strong back-and-forth barbed bisimulation} is a symmetric relation on coherent processes \(\sbfbb\) such that if \( R \sbfbb S\), then
	\begin{align*}
		R \bwlts{i}{\tau} R' \implies \exists S' \text{ \st} S \bwlts{i}{\tau} S' \text{ and } R' \sbfbb S' \tag{Back} \label{arriere} \\
		R\fwlts{i}{\tau} R' \implies \exists S' \text{ \st} S \fwlts{i}{\tau} S'\text{ and } R' \sbfbb S' \tag{Forth} \label{avant}    \\
		R \downarrow_{a} \implies S \downarrow_{a}. \tag{Barbed} \label{barbe}                                                         
	\end{align*}
	We write \( R \sbfbc S\) and define the \emph{strong back-and-forth barbed congruence} if \( R \sbfbb S\) and for all context \(C[\cdot]\), \(C [\orig{R}] \sbfbb C[\orig{S}]\).
\end{definition}

\begin{lemma}
	\label{lem-orig-sbfbc}
	\(R \sbfbc S \implies \orig{R} \sbfbc \orig{S}\).\end{lemma}
The proof is straightforward. The converse does not hold as \(R\) and \(S\) can be any derivative of \(\orig{R}\).

\section{Configuration structures}
\label{conf-struct}

We use configuration structures~\cite{Winskel1982,Glabbeek2001} as a denotational semantics for processes.
We recall the definitions and the operations necessary to encode processes, and refer to Winskel's work for the proofs.

\begin{notations}
	Let $E$ be a set, $\subseteq$ be the usual set inclusion relation and $C$ be a family of subsets of $E$.
	For $X\subseteq C$ we say that $X$ is \emph{compatible} and write $X\uparrow^{\text{fin.}}$ if $\exists y\in C$ finite such that $\forall x\in X$, $x\subseteq y$.
\end{notations}

\begin{definition} 	\label{def:conf_str}
	A \emph{configuration structure} $\mem{E,C}$ is a set \(E\) and $C \subseteq \power (E)$ satisfying:
	\begin{align*}
		\forall x \in C, \forall e \in x, \exists z \in C \text{ finite \st }e\in z \text{ and } z \subseteq x \tag{finitness} \label{def-finitness}                                                            \\
		\forall x \in C, \forall e, e' \in x, \text{ if } e\neq e' \text{ then }\exists z \in C, z \subseteq x \text{ and } (e\in z \iff e'\notin z) \tag{coincidence freeness} \label{def-coincidence-freenes} \\
		\forall X \subseteq C \text{ and } X \uparrow^{\text{fin.}} \Rightarrow \cup X \in C\tag{finite completness}                                                                                            \\
		\forall x,y\in C, \text{ if } x\cup y\in C \text{ then } x\cap y\in C \tag{stability} \label{stability}                                                                                                 
	\end{align*}
	
	A \emph{labelled configuration structure} $\conf=\mem{E,C,\labl}$ is a configuration structure endowed with a \emph{labelling function} $\labl:E\to\labels$.
	All configurations structures from now on will be supposed to be labelled.
\end{definition}

The elements of $E$ are called \emph{events} and subsets of $C$ \emph{configurations}. Intuitively, events are the actions occurring during the run of a process, while a configuration represents a state reached at some point.
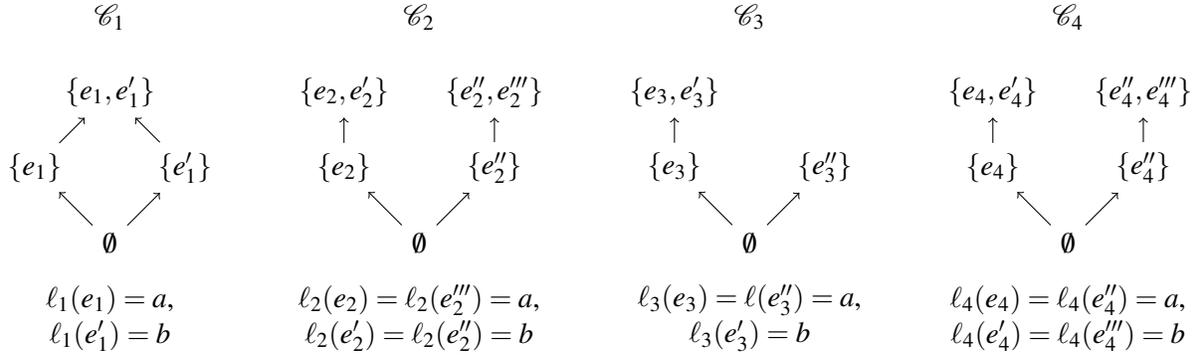
\begin{figure}
	\begin{tikzpicture}
		\node (nom) at (0, 2) {\(\conf_1\)};
		\node (emptyset) at (0, -1) {\(\emptyset\)};
		\node (a) at (-1, 0) {\(\{e_1\}\)};
		\node (b) at (1, 0) {\(\{e_1'\}\)};
		\node (ab) at (0, 1) {\(\{e_1, e_1'\}\)};
		\draw [->] (emptyset) -- (a);
		\draw [->] (emptyset) -- (b);
		\draw [->] (a) -- (ab);
		\draw [->] (b) -- (ab);
		\node[align=center] (labels) at (0, -2) {\(\labl_1(e_1) = a\),\\ \(\labl_1(e'_1) = b\)};
	\end{tikzpicture}
	\hfill
	\begin{tikzpicture}
		\node (nom) at (0, 2) {\(\conf_2\)};
		\node (emptyset) at (0, -1) {\(\emptyset\)};
		\node (a) at (-1, 0) {\(\{e_2\}\)};
		\node (b) at (1, 0) {\(\{e''_2\}\)};
		\node (ab) at (-1, 1) {\(\{e_2, e'_2\}\)};
		\node (ba) at (1, 1) {\(\{e''_2, e'''_2\}\)};
		\draw [->] (emptyset) -- (a);
		\draw [->] (emptyset) -- (b);
		\draw [->] (a) -- (ab);
		\draw [->] (b) -- (ba);
		\node[align=center] (labels) at (0, -2) {\(\labl_2(e_2) = \labl_2(e_2''') =a\),\\ \(\labl_2(e_2') = \labl_2(e_2'') =b\)};
	\end{tikzpicture}
	\hfill
	\begin{tikzpicture}
		\node (nom) at (0, 2) {\(\conf_3\)};
		\node (emptyset) at (0, -1) {\(\emptyset\)};
		\node (a) at (-1, 0) {\(\{e_3\}\)};
		\node (b) at (1, 0) {\(\{e_3''\}\)};
		\node (ab) at (-1, 1) {\(\{e_3, e'_3\}\)};
		\draw [->] (emptyset) -- (a);
		\draw [->] (emptyset) -- (b);
		\draw [->] (a) -- (ab);
		\node[align=center] (labels) at (0, -2) {\(\labl_3(e_3) = \labl(e_3'') = a\),\\ \(\labl_3(e'_3) = b\)};
	\end{tikzpicture}
	\hfill
	\begin{tikzpicture}
		\node (nom) at (0, 2) {\(\conf_4\)};
		\node (emptyset) at (0, -1) {\(\emptyset\)};
		\node (a) at (-1, 0) {\(\{e_4\}\)};
		\node (b) at (1, 0) {\(\{e''_4\}\)};
		\node (ab) at (-1, 1) {\(\{e_4, e'_4\}\)};
		\node (ba) at (1, 1) {\(\{e''_4, e'''_4\}\)};
		\draw [->] (emptyset) -- (a);
		\draw [->] (emptyset) -- (b);
		\draw [->] (a) -- (ab);
		\draw [->] (b) -- (ba);
		\node[align=center] (labels) at (0, -2) {\(\labl_4(e_4) = \labl_4(e_4'') =a\),\\\(\labl_4(e_4') = \labl_4(e_4''') =b\)};
	\end{tikzpicture}
	\caption{Four examples of configuration strutures}\label{ex_unif}
\end{figure}

\begin{example}
	In \autoref{ex_unif}, the configuration structure \(\conf_1\) have two events \(e_1\), \(e_1'\), with labels respectively $a$ and $b$, that are concurrent.
	Configuration \(\set{e_1}\) then corresponds to the process that fired action $a$.
	Its only possibility is then to fire \(b\) and reach the state $\set{e_1, e_1'}$.
	A process corresponding to this structure is \(a.0 | b.0\).
	The configuration structure \(\conf_2\) corresponds to a process where the events labelled respectively $a$ and $b$ are causally dependent, as in \(a.b.0+b.a.0\).
\end{example}

The configuration structure corresponding to a process $P$ is defined inductively on the syntax of $P$. Hence the encoding of a process is built from the encoding of its parts, unlike other models such as process graphs (or prime graphs) for CSSK~\cite{Phillips2007}. Moreover, configuration structures are \emph{compositional} in the sense that we can compose configuration structures into new structures.
Compositionality is an important feature as it allows us to reason on the context of a process.

Henceforth we detail how the operations of process algebras are translated on configuration structures, which in some cases have a nice categorical interpretation.
While the underlying category theory is not used in the paper, it can help in understanding how these structures behave.

\begin{definition}[Category of labelled configuration structures]
	A morphism of labelled configurations structures $f:\mem{E_1,C_1,\labl_1}\to\mem{E_2,C_2,\labl_2}$ is a partial function on the underlying sets $f:E_1\to E_2$ that is:
	\begin{align}
		\forall x\in C_1, f(x)=\set{f(e)~\vert~ e\in x}\in C_2 \tag{configuration preserving}           \\
		\forall x\in C_1, \forall e_1, e_2\in x, f(e_1)=f(e_2) \implies e_1=e_2 \tag{locally injective} \\
		\forall x \in C_1, \forall e \in x, \labl_1(e)=\labl_2(f(e)) \tag{label preserving}             
	\end{align}
	
	The configuration structures and their morphisms form a category.
\end{definition}

\begin{definition}[Operation on configuration structures {\cite{Winskel1982}}]
	\label{cat-op-def}
	We let $\conf_1=\mem{E_1,C_1,\labl_1}$, $\conf_2=\mem{E_2,C_2,\labl_2}$ be two configuration structures, set \(E^\star=E \cup \set{\star}\) and define the following operations:

		\paragraph{Product} \label{defproduct}
		      Define \emph{the product of \(\conf_1\) and \(\conf_2\)} as $\conf=\conf_1\times\conf_2$, for $\conf=\mem{E,C,\labl}$, where $E = E_1^\star\times E_2^\star$ is the product on sets with the projections $\pi_1$, $\pi_2$ and
		      \[
		      	x \in C \iff
		      	\begin{dcases*}
		      		\pi_1(x) \in C_1\text{ and }\pi_2(x)\in C_2,                                   \\
		      		\pi_1: \conf\to\conf_1\text{ and }\pi_2:\conf\to\conf_2 \text{ are morphisms}, \\
		      		x \text{ satisfies \eqref{def-finitness} and \eqref{def-coincidence-freenes}.} 
		      	\end{dcases*}
		      \]
		      The labelling function $\labl$ is defined as $\labl(e)=(\labl_1(e_1),\labl_2(e_2))$, where $\pi_1(e)=e_1$ and $\pi_2(e)=e_2$.
		      
		\paragraph{Coproduct}\label{defcoproduct}
		      Define \emph{the coproduct of \(\conf_1\) and \(\conf_2\)} as $\conf=\conf_1+\conf_2$, for $\conf=\mem{E,C,\labl}$, where $E=(\set{1}\times E_1)\cup(\set{2}\times E_2)$ and $C=\set{\set{1}\times x~|~x\in C_1}\cup\set{\set{2}\times x~|~x\in C_2}$.
		      The labelling function $\labl$ is defined as $\labl(e)=\labl_i(e_i)$ when $e_i\in E_i$ and $\pi_i(e_i)=e$.
		      
		\paragraph{Restriction}
		      Let $E'\subseteq E$.
		      Define \emph{the restriction of a set of events} as $\mem{E,C,\labl}\restr E'=\mem{E',C',\labl'}$ where $x'\in C'\iff x\in C, x\subseteq E'$. 		      \emph{The restriction of a name} is then $\mem{E,C,\labl}\restr E_a$ where $E_a=\set{e\in E~|~\labl(e)\neq\tau, a\in\labl(e)}$.
		      
		\paragraph{Prefix}
		      Define \emph{the prefix operation on configuration structures} as $\alpha.\mem{E,C,\labl}=\mem{e\cup E,C',\labl'}$, for $e\notin E$ where $x'\in C' \iff\exists x\in C, x'=x\cup e$ and \(\labl'(e) = \alpha\), and \(\forall e' \neq e\), \(\labl'(e') = \labl(e')\).
		      
		\paragraph{Relabelling}\label{relabel}
		      Define \emph{the relabelling of a configuration structure} as $\conf_1\circ\labl=\mem{E_1,C_1,\labl_1\circ\labl}$, where $\labl$ is a labelling function.
		      
		\paragraph{Parallel composition}\label{par-comp}
		      Define $\conf_1\|\conf_2=\big((\conf_1\times\conf_2)\circ\labl\big)\restr E$ where $\labl$ is defined as follows
		      \begin{align*}
		      	\labl(a)=a & \quad\labl(\tau)=\tau & \labl(a,\out a)=\labl(\out a, a)=\tau & \quad 
		      	\labl(\tau,a)=\labl(a,\tau)=0&\quad\labl(a,\out b)=\labl(\out b,a)=0
		      \end{align*}
		      and for $(\conf_1\times\conf_2)\circ\labl=\mem{E',C',\labl'}$ we have the set $E=\set{e\in E'~|~\labl'(e)\neq 0}$.

\end{definition}

In configuration structures $\mem{E,C,\labl}$ we denote $x\redl{e}x'$ the configurations $x,x'\in C$ such that $x=x'\cup\set{e}$ and with $x'\revredl{e}x$ the symmetric relation.
We use $x'\fbl{e} x$ for either $x\redl{e}x'$ or $x'\revredl{e}x$: if \(\labl(e) = \alpha\), we sometimes write $x \fbl{\alpha} x'$.

\begin{definition}[Partial order]
	\label{def:causality}
	Let $x\in C$ and $e_1,e_2\in x$. Then
	$e_1\leq_x e_2$ iff $\forall x_2\in C, x_2\subseteq x, e_2\in x_2\implies e_1\in x_2$.
\end{definition}
If $e_1\leq_x e_2$, we say that $e_1$ \emph{happens before} $e_2$ or that $e_1$ causes $e_2$ in the configuration $x$.
Morphisms on configuration structures reflect causality:
if $\pi:\conf_1\to\conf_2$ and for $e_1,e_2\in x$ and $x\in C_1$, if $\pi(e_1) \leq_{\pi(x)}\pi(e_2)$ then $e_1\leq_x e_2$.
\begin{definition}[Substructure]
	\label{def:inclusion}
	$\mem{E_1,C_1,\labl_1}\subseteq \mem{E_2,C_2,\labl_2}$ iff $E_1\subseteq E_2, C_1\subseteq C_2 \mbox{ and }\labl_1=\labl_2|_{E_1}$.
\end{definition}

\section{Encoding RCCS in configuration structures}
\label{RCCS-in-conf}

We start by encoding a CCS term into configuration structures and show an operational correspondence between the term and its encoding.
Intuitively, the configuration structure of a process without memory depicts all its possible future behaviour.
We also introduce a notion of context for configuration structures. 
Then we proceed to encode a RCCS term. A reversible process can do backward transitions but only up to a point: until it reaches the empty memory. We encode then a RCCS terms as an \enquote{address} in the configuration structure of its origin. This allows us to encode both the past and the future of a process in the same configuration structure. However the syntax of a process is not informative enough, hence we restrict the encoding to a class of processes.
Lastly we show an operational correspondence for RCCS terms and their encoding.

\subsection{Encoding CCS}
We start by encoding a term without memory, that is a CCS term. We do so by structural induction on the term using the operations defined previously (\autoref{cat-op-def}):
\begin{align*}
	\enc{P_1\vert P_2} & = \enc{P_1}\vert\enc{P_2} & \quad \enc{P_1+ P_2} & = \enc{P_1}+ \enc{P_2} & \quad \enc{\alpha.P} & = \alpha.\enc{P} & \quad \enc{\nu a.P} & = \enc{P}\restr E_a 
\end{align*}
Note that this encoding and its correspondence with CCS was first proposed by Winksell~\cite{Winskel1982}.

To show a strong bisimulation between a CCS process and its encoding, we introduce the following transformation of a configuration structures representing, intuitively, the structure we obtain after a transition: $\mem{E,C,\labl}\setminus x = \mem{E',C',\labl\restr E'}\text{ with }E'=\cup C'\text{ and }x'\in C'\iff\exists y\in C, x \subseteq y\text{ and }x'=y\setminus x.$

Intuitively, \(\conf\setminus x\) is the configuration resulting from the suppression of the events of \(x\) in all configurations of \(\conf\).
We call \emph{minimal} (with respect to the partial order in \autoref{def:causality}) an event whose singleton is a configuration.

\begin{proposition}
	Let \(x\) be a configuration in \(\conf\), then $\conf \setminus x$ is a configuration structure.
\end{proposition}

\begin{proposition}[Strong bisimulation between a CCS process $P$ and $\encc{P}$]
	\label{prop:encode_ccs}
	\begin{align}
		\text{If } P\redl{\alpha}Q\text{ then }\exists e\in\encc{P}\text{ minimal \st }\labl(e)=\alpha\text{ and }\encc{Q}=\encc{P}\setminus \set{e} \tag{Soundness} \label{soundness} \\
		\forall e\in\encc{P} \text{ minimal, }\exists Q\text{ \st }P\redl{\labl(e)}Q\text{ and } \encc{Q}=\encc{P}\setminus \set{e}\tag{Completeness} \label{completeness}             
	\end{align}
\end{proposition}

\begin{proof}
	We show this by induction on the derivation $P\redl{\alpha}Q$ for \eqref{soundness} and by structural induction on $\encc{P}$ for \eqref{completeness}.
\end{proof}

We cannot define a notion of context for configuration structures in general, as it is not clear what a configuration structure with a hole would be.
However, if a configuration structure $\conf$ has an operational meaning, i.e.\ if $\exists P$ a CCS process such that $\conf=\encc{P}$, we can use a CCS context \(C[\cdot]\) that we instantiate with $P$.

When reasoning on contexts in CCS, it is common to distinguish between the part of a transition fired by the context alone and the part fired by the process.
In the operational setting, one can easily decompose the term $C[P]$ thanks to the rules of the LTS.
We need a similar reasoning for the term $\encc{C[P]}$, hence we attach to the context $C[\cdot]$ and process $P$ a projection morphism \(\pi_{C, P}:\encc{C[P]}\to\encc{P}\) that can retrieve the parts of a configuration in $\encc{C[P]}$ that belong to $\encc{P}$\footnote{The formal definitions and the missing proofs can be found in \autoref{appendix}.}.

Morphisms do not preserve causality in general.
In the case of a product we can show that all causalities are due to one of the two configuration structures.
\begin{proposition}
	\label{prop:cause_projection}
	Let $x\in\conf_1\times\conf_2$. Then $e<_x e'\iff$ either $\pi_1(e)<_{\pi_1(x)} \pi_1(e')$ or $\pi_2(e)<_{\pi_2(x)} \pi_2(e')$.
\end{proposition}
Without much difficulty the result can be extended to say that in $\encc{C[P]}$, causality appears due to either the causality in $C[\cdot]$ or the causality in \(P\):
a context can add but cannot remove causality in the process~\cite{Cristescu2013}.
\subsection{Encoding RCCS}

A RCCS term corresponds to a configuration in the configuration structure of its origin.
We can use the past execution, that is the memory of $R$ to point to a configuration but it is not discriminatory enough.
Consider the process $\emptymem\rhd a.0+a.b.0\redl{a}R$ whose configuration structure is \(\conf_3\) in \autoref{ex_unif}.
To determine which of the configurations labelled $a$ correspond to $R$ we have to consider the future of $R$ as well.

Hence we choose a configuration that respects the past and the future of $R$, but this is still not enough.
Consider the process $a.b.0 + a.b.0$ whose configuration is \(\conf_2\) in \autoref{ex_unif}.
For the trace $\emptymem \rhd a.b+a.b \redl{\alpha} b$ there is no way to choose between the two configurations labelled $a$.
From now on, we consider only RCCS processes for which the underlying CCS process has the property that $\collapse(P)=P$, where $\collapse$ is defined below.

\begin{definition}[Collapse]
	\label{def:collapse}
	\begin{align*}
		\collapse(\alpha.P+\alpha.Q)= & \alpha. \collapse(P),\text{ if }\collapse(P) = \collapse(Q) &   
		\collapse(\alpha.P)=&\alpha.\collapse(P)\\
		\collapse(\alpha.P+\beta.Q)=  & \alpha.\collapse(P)+\beta.\collapse(Q)                      &   
		\collapse((a)P)=&(a)\collapse(P)\\
		\collapse(\alpha.P|\alpha.Q)= & \alpha. \collapse(P),\text{ if }\collapse(P) = \collapse(Q) &   
		\collapse(0) = 0\\
		\collapse(P|Q)=&\collapse(P)|\collapse(Q)
	\end{align*}
\end{definition}

Hence each process points to a unique configuration, enabling us to encode the past behaviour without difficulty.
Thus we define an \enquote{address} function that, given the configuration structure of the process's origin and a trace to the process we want to encode, returns the configuration corresponding to the current state.

\begin{definition}[Encoding RCCS processes in configuration structures]
	\label{def:encod_rccs}
	Given \(R\) a RCCS process, its encoding \(\encr{R}\) is defined as the couple \((\encc{\orig{R}}, \funaddress{\encc{\orig{R}}}{\emptyset}{\orig{R} \red^{\star} R})\), where
	\begin{align*}
		\funaddress{\encc{\orig{R}}}{x}{R_1\redl{\alpha}R_2\red^{\star} R_3} & =\funaddress{\encc{\orig{R}}}{x\cup\set{e}}{R_2\red^{\star} R_3)} &   & \text{if }         
		\begin{dcases*}
		x\cup\set{e}\in \encc{\orig{R}}\\
		\text{and}\\
		\enc{\erase(R_2)}\subseteq \big(\encc{\orig{R}}\setminus (x\cup\set{e})\big)
		\end{dcases*}\\
		\funaddress{\encc{\orig{R}}}{x}{R_2\red^{\star}R_3)}                 & = x                                                               &   & \text{if } R_2=R_3 
	\end{align*}
\end{definition}

Let us show that the encoding is correct, and in particular that the function $\address$ is well defined.
\begin{proposition}[Soundness of the RCCS encoding]
	\label{prop-soundness-rccs}
	Let $R$ be a process, then \(\exists !x \in \encc{\orig{R}}\) such that $\funaddress{\encc{\orig{R}}}{\emptyset}{\orig{R}\red^{\star} R}=x$.
\end{proposition}

The proof, presented in \autoref{appendix}, proceeds by induction on the trace, uses \autoref{prop:encode_ccs} and the collapsing hypothesis (\autoref{def:collapse}).

Let us now define a transition relation on configuration structures, useful in showing the operational correspondence between terms of RCCS and their encoding.

\begin{definition}[Transition in configuration structures]
	Define $(\encc{P}, x)\redl{\labl(e)}(\encc{P}, x\cup\set{e})$ for $x\cup\set{e}\in\encc{P}$.
\end{definition}

\begin{lemma}[Operational correspondence] 	\begin{enumerate}
		\item if $R\fbwlts{i}{\alpha}S$ then $\encr{R}\fbl{\alpha}\encr{S}$;
		\item let $\encr{R}=(\conf,x)$; if $(\conf,x)\redl{\labl(e)}(\conf, x\cup\set{e})$ then $\exists S$, such that for some $i\in\ids$ fresh, $R\fbwlts{i}{\alpha}S$ and $\encr{S}=(\conf,x\cup\set{e})$.
	\end{enumerate}
\end{lemma}

\begin{proof}
	\begin{enumerate}
		\item As $R\fwlts{i}{\alpha}S$, $\orig{R}=\orig{S}$ and we are in the following situation:
		      
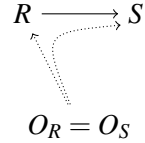
\tikz[baseline={([yshift=-.8ex]current bounding box.center)}]{
		      	\node (OR) at (0.75, 0) {\(\orig{R} = \orig{S}\)};
		      	\node (R) at (0, 1.5) {\(R\)};
		      	\node (S) at (1.5, 1.5) {\(S\)};
		      	\draw [->, densely dotted] (OR) -- (R);
		      	\draw [->, densely dotted] (OR) .. controls (0.25, 1.25) .. ($(S)-(0.2, .15)$);
		      	\draw [->] (R) -- (S);}
		      
		      We have that $\encr{S}=(\encc{\orig{R}},x_s)$, where
		      \(x_s= \funaddress{\encc{\orig{R}}}{\emptyset}{\orig{R}\red^\star S}=
		      \funaddress{\encc{\orig{R}}}{\emptyset}{\orig{R}\red^\star R\redl{\alpha} S}= x_R\cup\set{e}\).
		      As $\encr{R}=(\encc{\orig{R}},x_R)$ it follows that $(\encc{\orig{R}},x_R)\redl{\alpha}(\encc{\orig{R}},x_S)$.
		      The proof for the backward direction is similar except that it uses the trace up to $R$.
		\item From $(\conf,x)\redl{\labl(e)}(\conf, x\cup\set{e})$ we have that $x\cup\set{e}\in\conf$. Then $\set{e}\in\conf\setminus x$. From $\encr{R}=(\conf,x)$ we have that $\conf\setminus x=\encc{\erase(R)}$, hence $\set{e}\in\encc{\erase(R)}$. We use \autoref{prop:encode_ccs} and obtain that $\exists P$ such that $\erase(R)\redl{\labl(e)}P$.
		      Then due to the strong bisimulation between a RCCS term and its corresponding CCS term~\cite{Danos2004}, we have that, for some $i$, $R\fbwlts{i}{\alpha}S$.
		      That $\encr{S}=(\conf,x\cup\set{e})$ follows from a similar argument as in 1.\ above.\qedhere
	\end{enumerate}
\end{proof}

\section{Definition of Bisimulations}
\label{def-bisim-cs}

In this section we adapt to configuration structures the definitions of barb and strong back-and-forth barbed bissimulation on RCCS terms (\autoref{def:barb} and \autoref{def:sbfc_rccs}).
We define hereditary history preserving bisimulation, show that they \enquote{translate} the sister notion on RCCS terms (\autoref{soundness-bissim}), and use two family of relations, denoted $F_i$ and $B_i$, to inductively approximate the bisimulation (\autoref{Fn-Bn-and-h}).

\begin{definition}	\label{bisim-cs}
	A \emph{strong back-and-forth barbed bisimulation on labelled configuration structures} is a symmetric relation $\rel\subseteq C_1\times C_2$ such that \((\emptyset,\emptyset)\in\rel\), and if \((x_1,x_2)\in\rel\), then
	\begin{align*}
		x_1\revredl{e_1}x_1' \implies \exists x_2'\in C_2\text{ \st } x_2\revredl{e_2}x_2'\text{, with }\labl_1(e_1)=\labl_2(e_2)=\tau\text{ and } (x_1',x_2')\in\rel; \tag{Back}                                  \\
		x_1\redl{e_1}x_1' \implies \exists x_2'\in C_2 \text{ \st }x_2\redl{e_2}x_2'\text{, with }\labl_1(e_1)=\labl_2(e_2)=\tau \text{ and }(x_1',x_2')\in\rel; \tag{Forth}                                       \\
		\text{if }\exists e_1\in E_1\text{ \st }\labl_1(e_1)\neq\tau\text{ and }x_1\redl{e_1}x_1'\text{ then }\exists x_2'\in C_2\text{ \st }x_2\redl{e_2}x_2'\text{, with }\labl_1(e_1)=\labl_2(e_2).\tag{Barbed} 
	\end{align*}
	Let $\conf_1\sbbisim\conf_2$ if and only if there exists a strong back-and-forth barbed bisimulation between $\conf_1$ and $\conf_2$.
	
	Denote $\scbisim$ a symmetric relation on terms that have an operational meaning such that if $\enc{P_1}\scbisim\enc{P_2}$ then $\forall C$, $\enc{C[P_1]}\sbbisim\enc{C[P_2]}$.
\end{definition}

Let us now show that the relation in \autoref{bisim-cs} is the relation induced by the barbed congruence on processes.
\begin{lemma}
	\label{soundness-bissim}
	$R\sbfbc S\implies\encc{\erase(\orig{R})}\sbfbc\encc{\erase(\orig{S})}$ and $\encc{P}\sbfbc\encc{Q}\implies P\sbfbc Q$.
\end{lemma}

\begin{proof}
	Both case are similar :
	\begin{align*}
		R\sbfbc S & \implies \orig{R} \sbfbc \orig{S}                                                     &   & \text{(By \autoref{lem-orig-sbfbc})}                              \\
		          & \implies \erase(\orig{R}) \sbfbc \erase(\orig{S})                                     &   & \text{(As \(\erase(\emptymem \vartriangleright P) = P\))}         \\
		          & \implies \forall C[\cdot], C[\erase(\orig{R})] \sbfbb C[\erase(\orig{S})]             &   & \text{(By \autoref{def:sbfc_rccs})}                               \\
		          & \implies \forall C[\cdot], \enc{C[\erase(\orig{R})]} \sbfbb \enc{C[\erase(\orig{S})]} &   & \text{(By the \ref{soundness} part of \autoref{prop:encode_ccs})} \\
		          & \implies \encc{\erase(\orig{R})}\sbfbc\encc{\erase(\orig{S})}                         &   & \text{(By \autoref{bisim-cs})} \hfill \qedhere                    
	\end{align*}
	
\end{proof}
\begin{definition} 	\label{sbisim_h-def}
	A \emph{hereditary history preserving bisimulation on labelled configuration structures} is a symmetric relation $\rel\subseteq C_1\times C_2\times\power(E_1\times E_2)$ such that $(\emptyset,\emptyset,\emptyset)\in\rel$ and if $(x_1,x_2,f)\in\rel$, then
	\begin{align*}
		f\text{ is a label and order preserving bijection between }x_1\text{ and }x_2\notag                                                                \\
		x_1\redl{e_1}x_1'\implies \exists x_2'\in C_2 \text{\st } x_2\redl{e_2}x_2' \text{ and } f=f'\restr x_1, (x_1',x_2',f')\in\rel \notag              \\
		x_1 \revredl{e_1} x_1'\implies \exists x_2'\in C_2 \text{\st } x_2 \revredl{e_2} x_2'\text{ and } and f'=f\restr x_2, (x_1',x_2',f')\in\rel \notag 
	\end{align*}
	We define \emph{bisimilarity}, denoted $\conf_1\sbisim\conf_2$, as the greatest hereditary history preserving bisimulation on labelled configuration structures.
\end{definition}

Note that $\conf_1\sbisim\conf_2$ is an abuse of notation as $\sbisim$ is a relation defined on $C_1\times C_2\times\power(E_1\times E_2)$. Due to the restrictions imposed on the configuration structures (see \autoref{rem-concur}) there is a unique mapping between events for the greatest hhp bisimulation.

We can give an inductive characterisation of HHPB by reasoning on the structures up to a level: we ignore the configurations that have greater cardinality then the considered level. Hhpb is then the relation obtained when we reach the top level. Hence we can detect, whenever two configuration structures are not hhp bisimilar, at which level the bisimulation does no longer hold. We do this with the aid of the two following functions.

\begin{definition}[\(\Forw{i}\), \(\Backw{i}\)]
	\label{ForwBackwDef}
	Given \(\conf_1\), \(\conf_2\) two configuration structures, we let \(k\) be the cardinal of the largest configuration of \(\conf_1\)\footnote{All the configurations we manipulate here are finite. In an infinite setting, this bound can be viewed as a way to define an \enquote{up to \(k\) steps bisimulation}.} and define, for all \(x_1 \in C_1\), \(x_2 \in C_2\) and \(f\) a label and order-preserving function:
	\begin{align*}
		(x_1, x_2, f) \in \Forw{i}                                                                                                                                       & \Leftrightarrow                
		\begin{dcases*}
		\card(x_1) = \card(x_2) = i, \text{ \(f\) any label and order-preserving function}                                                                               & \text{if \(i = k\)} \\
		\forall x'_1, \exists x'_2, x_1 \redl{e_1} x'_1, x_2 \redl{e_1} x'_2 \text{ and } f=f'\restr x_1\text{ \st} (x'_1, x'_2, f') \in \Forw{i+1}                      & \text{elsewhere}    
		\end{dcases*}\\
		(x_1, x_2, f) \in \Backw{i}                                                                                                                                      & \Leftrightarrow           
		\begin{dcases*}
		(x_1, x_2, f) \in \Forw{i}                                                                                                                                       & \text{if \(i = 0\)} \\
		\forall x'_1, \exists x'_2, x_1 \revredl{e_1} x'_1, x_2 \revredl{e_1} x'_2 \text{ and }f'=f\restr x_2\text{ \st}(x'_1, x'_2, f') \in \Forw{i-1} \cap \Backw{i-1} & \text{elsewhere}    
		\end{dcases*}
	\end{align*}
\end{definition}

The relation \(\Backw{i}\) is built on top of \(\Forw{i}\): it \enquote{tests for the backward steps} all the couples that \enquote{passed the forward test}.
It should be remarked that, with this definition, \(\Backw{i} \subseteq \Forw{i}\), but, at the price of slight modifications, one could define \(\Forw{i}\) on top of \(\Backw{i}\).

\begin{example}
	\label{example2}
	Consider \(\conf_3\) and \(\conf_4\) of \autoref{ex_unif}, the relations $F_n$ are enough to discriminate them:
	\begin{align*}
		F_2 = & \big(\set{e_3, e_3'}, \set{e_4, e_4'}\big); \big(\set{e_3, e_3'}, \set{e_4'',e_4'''}\big) &   
		F_1 = & \big(\set{e_3}, \set{e_4}\big); \big(\set{e_3}, \set{e_4''}\big)                          &   
		F_0 = &\emptyset
	\end{align*}
	This intuitively is due to the fact that forward transitions are enough to discriminate $a+a.b$ and $a.b+a.b$.
	However for comparing the processes $a~|~b$ and $a.b+b.a$ whose configurations are \(\conf_1\) and \(\conf_2\) of \autoref{ex_unif}, we need the backward moves as well:
	\begin{align*}
		F_2 = & \big(\set{e_1,e_1'},\set{e_2,e_2'}\big); \big(\set{e_1,e_1'};\set{e_2'',e_2'''}\big) &   
		F_1 = & \big(\set{e_1},\set{e_2}\big); \big(\set{e_1'};\set{e_2''}\big)                      &   
		F_0 = &\big(\emptyset,\emptyset\big)\\
		B_2=  & \emptyset                                                                            &   
		B_1=  & \big(\set{e_1},\set{e_2}\big); \big(\set{e_1'};\set{e_2''}\big)                      &   
		B_0=&F_0=\big(\emptyset,\emptyset\big)
	\end{align*}
\end{example}
The following proposition states that pairs of configurations are in a bisimulation relation if they have the same cardinality. It follows from the fact that any configuration is reachable from the empty set and that they have to mimic each other's step in the backward direction.
\begin{proposition}
	\label{prop:size_config}
	Let $\conf_1\sbisim\conf_2$ and $x_1\in\conf_1$, $x_2\in\conf_2$.
	If $\exists f$ such that $(x_1,x_2,f)\in \{\sbisim\}$ then $\card(x_1)=\card(x_2)$.
\end{proposition}

We are going to prove a fine lemma that will be handy to prove \autoref{main-thm}.
It implies that if for all \(n \leqslant k\) the maximal cardinal considered, \(\Forw{n} \cap \Backw{n} \neq \emptyset\), then \(\cup_{n \leqslant k} (\Forw{n} \cap \Backw{n})\) is a bisimulation.

\begin{lemma}
	\label{Fn-Bn-and-h}
	For all \(\conf_1\), \(\conf_2\), if \(\conf_1 \sbisim \conf_2\), then \(\forall x_1 \in \conf_1 (\exists x_2 \in \conf_2, \exists f, (x_1, x_2, f) \in \Forw{n} \cap \Backw{n}) \iff (\exists x_2 \in \conf_2, \exists f, (x_1, x_2, f) \in\sbisim )\).
\end{lemma}

\begin{proof}
	Let us denote $\rel$ the relation $\sbisim$.
	One should first remark that \(\conf_1 \sbisim \conf_2\) implies that \(\forall x_1 \in \conf_1, \exists x_2 \in \conf_2\), and $\exists f$ such that \( (x_1, x_2, f) \in \rel\), as $(\emptyset,\emptyset,\emptyset)\in\rel$ and all configurations are reachable from the empty set.
	The reader should notice that the \(x_2 \in \conf_2\) and \(f\) on both sides of the \(\iff\) symbols may be different.
	
	We prove that statement by induction on the cardinal of \(x_1\).
	\paragraph{\(\card(x_1) = 0\)}
	\begin{itemize}
		\item[\(\Rightarrow\)]
		      \(x_2 \in \conf_2\) \st \((\emptyset, x_2, f) \in \rel\) follows by the definition of the bisimulation from \(x_2 = \emptyset\) and \(f = \emptyset\).
		\item[\(\Leftarrow\)]
		      By definition, \(\Forw{0} \cap \Backw{0} = \Forw{0}\).
		      Since there exists \(x_2 \in \conf_2\) such that \((\emptyset, x_2, f) \in \rel\), we know that any forward transition made by \(\emptyset\) can be simulated by a forward transition from \(x_2\), and that the elements obtained are in the relation \(\rel\).
		      By an iterated use of this notion, we can find \enquote{maximal} elements \(x_1^m \in \conf_1\) and \(x_2^m \in \conf_2\) (that is, elements of maximal cardinality, \(k\)) such that \((x_1^m, x_2^m, f^m) \in \rel\).
		      By \autoref{prop:size_config}, \(x_1^m\) and \(x_2^m\) have the same cardinality, and \((x_1^m, x_2^m, f^m) \in \Forw{k}\).
		      By just \enquote{reversing the trace}, we can go backward and stay in relation \(\Forw{i}\) until \(i = 0\), hence we found the \(x_2\) and \(f\) we were looking for.
	\end{itemize}
	\paragraph{\(\card(x_1) = k +1\)}
	As \(\card(x_1) > 0\), we know there exists \(x'_1\) such that \(x_1 \revredl{e_1} x'_1\).
	\begin{itemize}
		\item[\(\Rightarrow\)]
		      Let \(x_2\) and \(f\) such that \((x_1, x_2, f) \in \Forw{k+1} \cap \Backw{k+1}\).
		      We know that
		      \begin{align*}
		      	\forall x'_1, \exists x'_2\text{ and } f', x_1 \revredl{e_1} x'_1, x_2 \revredl{e_1} x'_2 \text{ and } (x'_1, x_2', f') \in \Backw{k} \tag{By Definition of \(\Backw{k} \)} \\
		      	\exists x''_2,f'', (x'_1, x_2'', f'') \in \rel \tag{By Induction Hypothesis}                                                                                                
		      \end{align*}
		      And as \(x'_1 \redl{e_1} x_1\), there exists \(x'''_2\) and \(f'''\) such that \((x_1, x'''_2, f''') \in \rel\).
		\item[\(\Leftarrow\)]
		      We prove it by contraposition: suppose that \(\exists x_2,f\) such that \((x_1, x_2, f) \in \rel\), we prove that \(\forall x_2\), \((x_1, x_2, f) \notin \Forw{k+1} \cap \Backw{k+1}\) leads to a contradiction.
		      
		      As \((x_1, x_2, f) \in \rel\), \(\exists x'_1, x'_2,f'\) such that \(x_1 \revredl{e_1} x'_1\), \(x_2 \revredl{e_1} x'_2\) and \((x'_1, x'_2, f') \in \rel\).
		      By induction hypothesis, \(\exists x''_2\) and \(\exists f''\) such that \((x'_1, x''_2, f'') \in \Forw{k} \cap \Backw{k}\).
		      As \(x'_1 \redl{e_1} x_1\), \(\exists x'''_2\) and \(\exists f'''\) such that \(x''_2 \redl{e_1} x'''_2\) and \((x_1, x'''_2, f''') \in \Forw{k+1}\), by definition of \(\Forw{k}\).
		      
		      So \((x_1, x'''_2, f''') \notin \Backw{k+1}\), but as \(x_1 \revredl{e_1} x'_1\) and \(x'''_2 \revredl{e_1} x''_2\), and as moreover \((x'_1, x''_2, f'') \in \Forw{k} \cap \Backw{k}\), we have that \((x_1, x'''_2, f''') \in \Backw{k+1}\).
		      
		      From this contradiction we know that we found the right element (\(x'''_2\)) that is in relation with \(x_1\) according to \(\Forw{k+1} \cap \Backw{k+1}\).
		      \qedhere
	\end{itemize}
\end{proof}

\section{Correspondence between HHPB and Strong Barbed Congruence}
\label{sec:res}

In this section we use the relations defined in \autoref{def-bisim-cs} to show that two processes are barbed congruent whenever their denotations are in the HHPB relation (\autoref{main-thm}).
One direction is straightforward (\autoref{prop:HHPB_congr}), whereas the other is more technical and, as in CCS~\cite{Milner1992}, follows by contradiction.
It uses the relations \(\Forw{i}\) and \(\Backw{i}\) (\autoref{ForwBackwDef}) to build contexts that discriminate processes that are not bisimilar.

\begin{remark}[On auto-concurrency and others limitations]
	\label{rem-concur}
	In the proofs that follow we need to uniquely identify configurations based on the labels and orders of the \enquote{open} (i.e.\ non synchronized) events.
	This is not possible in processes as $a.P~|~a.Q$ or $a.P+a.Q$.
	Auto concurrency~\cite[Definition 9.5]{Glabbeek2001} forbids the first kind of processes. But we need a stronger condition, a sort of auto conflict, to forbid the second, that is not ruled out by the collapse function (\autoref{def:collapse}).
	Hence in the following we do not consider processes that exhibit auto concurrency \emph{or} auto conflict.
	
	The problem is specific to the encoding in configuration structures.
	It appears in the encoding of Winskel~\cite{Winskel1982}, and is treated thanks to \emph{tags} that discriminates between the right- and the left-hand side of the sum and of the product~\cite{Winskel1995}.
	Hence we can retrieve the whole class of processes by adding more information on the labels, at the cost of a more cumbersome presentation.
	
\end{remark}

\begin{proposition}	\label{prop:HHPB_congr}
	\(\enc{P_1} \sbisim \enc{P_2} \implies \forall C, \enc{C[P_1]} \sbisim \enc{C[P_2]}
	\)
\end{proposition}

The proof, exposed in \autoref{appendix}, amounts to carefully build a relation between \(\enc{C[P_1]}\) and \(\enc{C[P_2]}\) that reflects the known bissimulation between \(\enc{P_1}\) and \(\enc{P_2}\).
Its uses that causality in a product is the result of the entanglement of the causality of its elements (\autoref{prop:cause_projection}).

\begin{theorem}
	\label{main-thm}
	\(\enc{P_1} \sbisim \enc{P_2} \iff \enc{P_1} \sbfbc \enc{P_2}\)
\end{theorem}

\newcommand{\horiz}{5} \pgfmathsetmacro{\horizhoriz}{2*\horiz}

\begin{figure}
	For \(i \in \{1, 2\}\), we have:
	
	\begin{tikzpicture}
		\fill[colorp1] (0,0) ellipse (1cm and 2cm) node[below=1.5cm, left=.7cm, black]{\(\enc{P_i}\)};
		\draw (0, -1) node[below]{\(x_i\)} node (x1) {$\bullet$};
		\draw (0, 1) node[above]{\(x'_i\)} node (x'1) {$\bullet$};
		\draw [->] (x1) -- (x'1);
		\fill[colorp1] (\horiz,0) ellipse (1cm and 2cm) node[below=1.5cm, left=.7cm, black]{\(\enc{C[P_i]}\)};
		\draw (\horiz, -1) node[below]{\(y_i\)} node (y1) {$\bullet$};
		\draw (\horiz, 1) node[above]{\(y'_i\)} node (y'1) {$\bullet$};
		\draw [->] (y1) -- (y'1);
		\fill[colorp1] (\horizhoriz,0) ellipse (1cm and 2cm) node[below=1.5cm, left=.7cm, black]{\(\enc{C'[P_i]}\)};
		\draw (\horizhoriz, 1) node[above]{\(z'_i\)} node (z'1) {$\bullet$};
		\draw [->, double] (y1) to [bend right] node[midway,above,sloped] {\(\pi_{C, P_i}\)} (x1);
		\draw [->, double] (z'1) to [bend right] node[midway,above,sloped] {\(\pi_{C', C[P_i]}\)} (y'1);
		\draw [->, double] (y'1) to [bend right] node[midway,above,sloped] {\(\pi_{C, P_i}\)} (x'1);
	\end{tikzpicture}
	
	We start with \(y_1 \sbfbc y_2\), then prove that \(z'_1 \sbfbc z'_2\), to end up with \((x'_1, x'_2, f) \in \Forw{n} \cap \Backw{n}\).
	
	\caption{Configurations Structures by the end of the proof of \autoref{main-thm}}
	\label{fig-main-thm}
\end{figure}

\begin{proof}
	The left-to-right direction follows from the definition of \(\sbisim\) (\autoref{sbisim_h-def}) and from \autoref{prop:HHPB_congr}.
	
	We prove the other direction 	by contraposition: let us suppose that \(\enc{P_1}\sbfbc\enc{P_2}\) and $\enc{P_1} \not\sbisim \enc{P_2}$, we will find a contradiction.
	\autoref{fig-main-thm} presents the general shape of the configurations at the end of the proof.
	
	As $\enc{P_1} \not\sbisim \enc{P_2}$, by \autoref{Fn-Bn-and-h}, there exists $x_1\in\encc{P_1}$ such that $\forall x_2\in\encc{P_2}$, $(x_1,x_2,f)\notin F_n\cap B_n$ holds.
	Note that we can only consider $x_2$ such that $\card(x_1)=\card(x_2)=n$, and that we use the projections \(\pi_{C, P}\) (\autoref{def-projection}) to separate the events of the process \(P\) from the events of the context \(C\).
	
	Let us show that for any $x_1$ we can define
	$ C[\cdot] := \prod_{e_i \in x_i} (\overline{\labl(e_i)} + c_{e_i}) | [\cdot] $
	where $c_{e_i}\notin\names(P_1)\cup\names(P_2)$, such that the following holds
	\begin{itemize}
		\item \(\exists y_1 \in \enc{C[P_1]}\) such that \(y_1\) is closed, \(\pi_{C, P_1} (y_1) = x_1\) and \(y_1 \not\downarrow{c_{e_i}}\) for all \(e_i \in x_1\);
		\item We supposed that \(\enc{P_1}\sbfbc \enc{P_2}\), so \(\enc{C[P_1]}{\sbisim}^{\tau}\enc{C[P_2]}\).
		      Hence \(\exists y_2 \in \enc{C[P_2]}\) such that \((y_1, y_2, g) \in {\sbisim}^{\tau}\) and \(y_2 \not\downarrow{c_{e_i}}\) for all \(e_i \in x_1\).
	\end{itemize}
	Moreover we show that $(x_1,\pi_{C, P_1}(y_2),f)\in F_n$, for some $f$ a label and order preserving bijection.
	
	Let us start by showing that such an $f$ exists.
	
	We denote $\pi_{C, P_1}(y_2)$ with $x_2$.
	We have that \(\forall e_1, e'_1 \in x_1\), and \(e_2 \in x_2\),
	\begin{equation}
		\label{reason2}
		e_2 \in x_2 \iff e_1 \in x_2\text{ and } \labl(e_1) = \labl (e_2)
	\end{equation}
	\begin{align}
		e_1 <_{x_1} e'_1 & \Longrightarrow \pi_{C, P_1}^{-1} (e_1) <_{y_1} \pi_{C, P_1}^{-1} (e'_1) \label{reason}        \\
		                 & \Longrightarrow g(\pi_{C, P_1}^{-1} (e_1)) <_{y_2} g(\pi_{C, P_1}^{-1} (e'_1)) \label{reason3} 
	\end{align}
	
	Remark that \eqref{reason2} follows from \(y_2 \not\downarrow{c_{e_i}}\) and from the fact that if $y_1$ is closed we can show by contradiction that $y_2$ is closed as well.
	Secondly, \eqref{reason} follows from \autoref{prop:cause_projection} and from the form of the context, which does not induce any causality between the events.
	Lastly, \eqref{reason3} follows from $g$ being an order preserving bijection between $y_1$ and $y_2$.
	
	We proceed by induction to show that \((x_1, x_2, f) \in \Forw{n}\).
	\begin{itemize}
		\item If \(n = k\)  for \(k\) the maximal cardinal of events in \(\enc{P_1}\).
		      This case is trivial, as \(\card(x_1) = \card(x_2) = k\).
		\item If \(n = k - 1\) for \(k>1\), we prove that \((x_1, x_2, f) \notin \Forw{k-1}\) leads to a contradiction.
		      There are two cases:
		      \begin{align}
		      	\not\exists x'_1, x_1 \redl{e_1} x'_1, \exists x'_2, x_2 \redl{e_2} x'_2 \label{case-no-transition}                                         \\
		      	\exists x'_1, x_1 \redl{e_1} x'_1, \forall x'_2, x_2 \redl{e_2} x'_2 \text{ and } (x'_1, x'_2, f') \notin \Forw{k}\label{case-no-extension} 
		      \end{align}
		      
		      The implication \eqref{case-no-transition} is easier: if \(\exists x'_2, x_2 \redl{e_2} x'_2\), then, as a context cannot remove transitions from the original process, \(\exists y'_2, y_2 \redl{(e_2, \star)} y'_2\).
		      As \(\enc{C[P_2]} \sbfbb \enc{C[P_1]}\), \(\exists y'_1, y_1 \redl{(e_1, \star)} y'_1\), and a similar argument on the context shows that \(\exists x'_1, x_1 \redl{e_1} x'_1\). Hence a contradiction.
		      
		      To prove \eqref{case-no-extension} requires more work and uses the induction hypothesis.
		      First, let
		      $C'[\cdot] := C[\cdot] | (\overline{\labl(e_1)} + c_{e_1})$.
		      By induction hypothesis, there exists \(z'_1 \in \enc{C'[P_1]}\) such that \(z'_1\) is closed, \(\pi_{C', C[P_1]} (z'_1) = y'_1\) and \(z'_1 \not\downarrow{c_{e_i}}\) and \(z'_1 \not\downarrow{c_{e_1}}\) for all \(e_i \in x_1\).
		      
		      By hypothesis, \(\enc{P_1}\sbfbc \enc{P_2}\), hence \(\enc{C'[P_1]}\sbfbb \enc{C'[P_2]}\) implies that $\exists z_2', h'$ such that \(z_2 \in \enc{C'[P_2]}\) and \(z'_2 \not\downarrow{c_{e_i}}\) and \(z'_2 \not\downarrow{c_{e_1}}\) for all \(e_i \in x_1\).
		      
		      Let us denote the projection \(\pi_{C', C[P_2]} (z'_2)\) as $y_2''$. As $z_1'$ is closed, so is $z_2'$. We can infer using the fact that $z_2'$ is closed and that \(z'_2 \not\downarrow{c_{e_1}}\) that $\exists e_2''\in y_2''$ such that $\labl(e_2'')=\labl(e_1)$ and $y_2''\setminus\set{e_2''}$ is closed.
		      
		      From \(z'_2 \not\downarrow{c_{e_i}}\) we have that \(y''_2 \not\downarrow{c_{e_i}}\).
		      As there exists a label and order preserving bijection $h'$ between $z_1'$ and $z_2'$, and as we forbid auto concurrency and \enquote{ambiguous} non deterministic sum (\autoref{rem-concur}), we conclude that $\pi_{C, P_2} (y''_2)=x_2'$ and \(\pi_{C', P_2} (z'_2)=x_2'\).
		      
		      Then we have \(\pi_{C', P_1}(z'_1) = x'_1, \pi_{C', P_2}(z'_2) = x'_2\), and by induction hypothesis, \( (x'_1, x'_2, f) \in \Forw{k}\).
		      But as \(x_1 \redl{e_1} x'_1\) and \(x_2 \redl{e_2} x'_2\), we have that \( (x_1, x_2, f) \in \Forw{k-1}\), hence a contradiction.
	\end{itemize}	
	To prove that \((x_1, x_2, f) \in \Backw{n}\), we use induction, the base case (\( n = 0\)) being trivial.
	The step case goes along the line of (and uses) the proof that \((x_1, x_2, f) \in \Forw{n}\).
\end{proof}

\section*{Conclusion}
\addcontentsline{toc}{section}{Conclusion}
We showed that, for a restricted class of RCCS processes (without recursion, auto-concurrency or auto-conflict) hereditary history preserving bisimilation has a contextual characterisation in CCS.
We used the barbed congruence defined on RCCS as the congruence of reference, adapted it to configuration structures and then showed a correspondence with HHPB.
As a proof tool, we defined two inductively relations that approximate HHPB.
Consequently we have that adding reversibility into the syntax helps in retrieving some of the discriminating power of configuration structures.

This work follows notable efforts~\cite{Phillips2007,Lanese2010}
to understand
equivalences for reversible processes.
There are many interesting continuations.
A first one as suggested in the introduction, is to move to weak equivalences, which ignore silent moves \(\tau\) and focus on the observable part of a process.
This is arguably a more interesting relation than the strong one, in which processes have to mimic each other's silent moves.
Even if such a relation on configuration structures exists~\cite{Vogler1993} one still has to show that this is indeed the relation we expect.
In the denotational setting, the adjective \enquote{weak} has sometimes~\cite{Phillips2012,Glabbeek2001} a different meaning: it stands for the ability to change the label and order preserving bijection as the relation grows, to modify choices that were made before this step.

The relations defined so far simulate forward (resp.\ backward) transitions only with forward (resp.\ backward) transitions, and only consider \enquote{forward} barb.
Ignoring the direction of the transitions could introduce some fruitful liberality in the way processes can simulate each other.
Depending on the answer, \(a + \tau . b\) and \(a + b\) would be weakly bisimilar or not. Moreover one can also consider irreversible moves and understand what are the meaningful equivalences in the setting of transactions~\cite{Danos2005}.

Context---which plays a major part in these equivalences---raises questions on the memory handling of RCCS: what about context that could \enquote{fix the memory} of an incoherent process?
For instance, $C=\mem{1,a,0}\vartriangleright P'|[\cdot]$ and $P=\mem{1,\bar{a},0}\vartriangleright P''$ are incoherent, but $C[P]$ is coherent and can backtrack.

One can easily retrieve auto concurrency and auto conflict by tagging the transitions.
Bisimulations have then to consider the tags. Maybe of less interest but important for the generality of these results, one should include infinite processes as well.
This needs a rework of the relations in \autoref{ForwBackwDef} used to approximate the HHPB.

\paragraph{Acknowledgement}
We would like to thank D. Varacca and J. Krivine for the very useful discussions as well as the referee for his helpful remarks.

\bibliographystyle{eptcs}

\newpage
\appendix
\renewcommand{\thesection}{\Alph{section}}

\section{Appendices}
\label{appendix}
\addcontentsline{toc}{section}{Appendices}

\subsection{Additional Definition}
\begin{definition}[Context for configuration structures]
	\label{def-projection}
	
	Let $P$ a CCS a process and $C[\cdot]$ a context. Then \(\pi_{C, P}\) is as the projection morphism \(\pi_{C, P}:\encc{C[P]}\to\encc{P}\) defined inductively on the structure of $\encc{C[P]}$:
	\begin{itemize}
		\item $\pi_{C, P}:\enc{\alpha.C'[P]}\to\enc{P}$ is defined as $\pi_{C, P}(e)=\pi_{C', P}(e)$;
		\item $\pi_{C, P}:\enc{C'[P]|P'}\to\enc{P}$ is defined as $\pi_{C, P}(e)=\pi_{C', P}(\pi_1(e))$, where $\pi_1:\enc{C'[P]|P'}\to\enc{C'[P]}$ is the projection morphism defined by the product;
		\item $\pi_{C, P}:\enc{C'[P]+P'}\to\enc{P}$ is defined as $\pi_{C, P}(e)=\pi_{C', P}(\pi_1(e))$, where $\pi_1:\enc{C'[P]+P'}\to\enc{C'[P]}$ is the projection morphism defined by the coproduct;
		\item $\pi_{C, P}:\enc{(a)C'[P]}\to\enc{P}$ defined as $\pi_{C, P}(e)=\pi_{C', P}(e)$.
	\end{itemize}
\end{definition}
That the projection $\pi_{C, P}:\enc{C[P]}\to\enc{P}$ is a morphism follows by a simple case analysis.
We naturally extend \(\pi_{C, P}\) to configurations.

\subsection{Proof of \autoref{prop-soundness-rccs}}
\label{soundness-of-rccs}

\begin{proof}
	Without loss of generality, the trace $\orig{R}\red^{\star} R$ can be considered to be only forward: every reversible trace can be re-arranged as a succession of backward moves followed by a succession of forward moves~\cite[Lemma~10]{Danos2004}, but \(\orig{R}\) cannot go backward.
	We proceed by induction on the trace $\orig{R}\red^{\star} R$.
	Let \(\funaddress{\encc{\orig{R}}}{\emptyset}{\orig{R}\red^{\star} R}=x_n\), for \(x_n\in\encc{\orig{R}}\) and such that \(\encc{\erase(R)}=\encc{\erase(\orig{R})}\setminus x_n\).
	We have to show that
	\[\funaddress{\encc{\orig{R}}}{\emptyset}{\orig{R}\red^{\star} R\redl{a}R_{n+1}}=x_n\cup\set{e}\text{ and }x_n\cup\set{e}\in\encc{\orig{R}}\text{, with }\enc{\erase(R_{n+1})}=\encc{\orig{R}}\setminus (x_n\cup\set{e}).\]
	
	We have that $\funaddress{\encc{\orig{R}}}{\emptyset}{\orig{R}\red^{\star} R\redl{a}R_{n+1}}=\funaddress{\encc{\orig{R}}}{x_n}{R\redl{a}R_{n+1}}$ and that $\encc{\erase(R)}=\encc{\orig{R}}\setminus x_n$.
	We want to show that for $R\redl{\alpha}R_{n+1}$ ,$\exists!\set{e}\in\encc{\erase(R)}$ such that $\encc{\erase(R_{n+1})}=\encc{\erase(R)}\setminus \set{e}$.
	We consider only the case $\alpha=a$, the rest is similar.
	We rewrite $R\congru(b_1\dots b_n)(m_1\rhd a.P_1~|~P_2)$ and $R_{n+1}\congru(b_1\dots b_n)(m_1\rhd P_1~|~P_2)$ and hence $\erase(R)=(b_1\dots b_n)(a.P_1~|~P_2)$ and $\erase(R_{n+1})=(b_1\dots b_n)(P_1~|~P_2)$. We want to show that $\exists !e\in\encc{\orig{R}}\setminus x_n$ such that $\labl(e)=\alpha$ and
	\[\encc{\erase(R_{n+1})}=\encc{\erase(\orig{R})}\setminus (x\cup\set{e}).\]
	But $\encc{\erase(\orig{R})}\setminus (x\cup\set{e})=\encc{\erase(R)}\setminus \set{e}$. Hence it is enough to show that $\exists !e\in\encc{\erase(R)}$ such that $\labl(e)=\alpha$ and
	\[\encc{\erase(R_{n+1})}=\encc{\erase(R)}\setminus \set{e}\]
	which is equivalent to show that
	\[ \encc{(b_1\dots b_n)(P_1~|~P_2)}=\encc{(b_1\dots b_n)(a.P_1~|~P_2)}\setminus \set{e}.\]
	From \autoref{prop:encode_ccs} such an event exists and its uniqueness follows from the collapsing hypothesis (\autoref{def:collapse}).
	
	Let us prove that if $x\in\encc{(b_1\dots b_n)(P_1~|~P_2)}$ then $x\in\encc{(b_1\dots b_n)(a.P_1~|~P_2)}\setminus \set{e}$. The other direction is similar. Let us unfold the definition of the encoding.
	We have the following equalities:
	\begin{align*}
		\encc{(b_1\dots b_n)(P_1~|~P_2)}   & =(b_1\dots b_n)(\encc{P_1}\times\encc{P_2})\restr X   \\
		\encc{(b_1\dots b_n)(a.P_1~|~P_2)} & =(b_1\dots b_n)(a.\encc{P_1}\times\encc{P_2})\restr Y 
	\end{align*}
	If $x\in(b_1\dots b_n)(\encc{P_1}\times\encc{P_2})\restr X$ then
	\begin{equation}
		\label{eq:labl}
		\nexists e\in x, \labl(e)\in\set{b,\out b,0}.
	\end{equation}
	Hence $x\in(\encc{P_1}\times\encc{P_2})$. Let $\pi_1$, $\pi_2$ the two projections defined by the product. Then
	\begin{equation}
		\label{eq:proj}
		\pi_1(x)\in\encc{P_1}\text{ and }\pi_2(x)\in\enc{P_2}.
	\end{equation}
	As $\pi_1(x)\in\encc{P_1}$, and from the definition of $\encc{a.P_1}$ we have that $\exists e_1$, $\labl(e_1)=a$ and such that $\set{e_1}\cup \pi_1(x)\in a.\encc{P_1}$. From \autoref{eq:proj} we have that $\exists x_2\in a.\encc{P_1}\times\encc{P_2}$ such that $\pi_1(x_2)=\set{e_1}\cup \pi_1(x)$ and $\pi_2(x_2)=\pi_2(x)$. Hence $\exists !e$ such that $\pi_1(e)=e_1$, $\pi_2(e)=\star$ and $x_2=\set{e}\cup x$. From \autoref{eq:labl} we have that $x_2\in(b_1\dots b_n)(a.\encc{P_1}\times\encc{P_2})\restr Y$. From the definition of $\encc{\orig{R}}\setminus \set{e}$ we infer that if $x\cup\set{e}\in(b_1\dots b_n)(a.\encc{P_1}\times\encc{P_2})\restr Y$ then $x\in\encc{(b_1\dots b_n)(a.P_1~|~P_2)}\setminus \set{e}$.
	
	From $\encc{\erase(R)}=\encc{\orig{R}}\setminus x_n$, we have that $\forall y\in \encc{\erase(R)}$, $\exists y\cup x_n\in\encc{\orig{R}}$. In particular $x_n\cup\set{e}\in\encc{\orig{R}}$.
	Hence $\funaddress{\encc{\orig{R}}}{\emptyset}{\orig{R}\red^{\star} R\redl{a}R_{n+1}}=x_n\cup\set{e}$ with $\encc{\erase(R_{n+1})}=\encc{\orig{R}}\setminus (x_n\cup\set{e})$.
\end{proof}

\subsection{Proof of \autoref{prop:HHPB_congr}}
\label{proof:HHPB_congr}

\begin{figure}
	{\centering
		\begin{tikzpicture}
			\fill[colorp1] (0,0) ellipse (1cm and 2cm) node[above=2cm, black]{\(\enc{P_1} = \mem{E_1, C_1, \labl_1}\)};
			\draw (0, -1) node[below]{\(x_1\)} node (x1) {$\bullet$};
			\draw (0, 1) node (x'1) {$\bullet$};
			\draw [->] (x1) -- node[right]{\(e''_1\)} (x'1);
			\fill[colorp1] (6, 0) ellipse (1cm and 2cm) node[above=2cm, black]{\(\enc{P_1 | Q} = \mem{E'_1, C'_1, \labl'_1} = (\enc{P_1} \times \enc{Q}) \restr X_1 \)};
			\fill[colorp1] (6.7, 0) ellipse (1cm and 2cm);
			\draw (6.35, -1) node[below]{\(y_1\)} node (y1) {$\bullet$};
			\draw (6.35, 1) node[above]{\(y'_1\)} node (y'1) {$\bullet$};
			\draw [->] (y1) -- node[right]{\(e'' = (e''_1, e''_q)\)}(y'1);
			\begin{scope}[yshift=-4.3cm]
				\fill[colorp2] (0,0) ellipse (1cm and 2cm) node[below=2cm, black]{\(\enc{P_2} = \mem{E_2, C_2, \labl_2}\)};
				\draw (0, -1) node[below]{\(x_2\)} node (x2) {$\bullet$};
				\draw (0, 1) node[above]{\(x''_2\)} node (x'2) {$\bullet$};
				\draw [->] (x2) -- node[right]{\(e''_2\)} (x'2);
				\fill[colorp2] (6, 0) ellipse (1cm and 2cm) node[below=2cm, black]{\(\enc{P_2 | Q} = \mem{E'_2, C'_2, \labl'_2} = (\enc{P_2} \times \enc{Q}) \restr X_2\)};
				\fill[colorp2] (6.7, 0) ellipse (1cm and 2cm);
				\draw (6.35, -1) node[below]{\(y_2\)} node (y2) {$\bullet$};
				\draw (6.35, 1) node[above]{\(y'_2\)} node (y'2) {$\bullet$};
				\draw [->] (y2) -- node[right]{\(e'_2\)}(y'2);
			\end{scope}
			\draw [<->, double] (y2) to [bend left] node[midway, left] {\(f_c\)} ($(y1)-(.3, 0)$);
			\draw [<->, double] (x2) to [bend left] node[midway, left] {\(f_c\)} ($(x1)-(.3, 0)$);
			\draw [->, double] (y2) to [bend right] node[midway, below ,sloped] {\(\pi_2\)} (x2);
			\draw [->, double] (y1) to [bend right] node[midway, below ,sloped] {\(\pi_1\)} (x1);
			\draw [->, double] (y'1) to [bend right] node[midway, below ,sloped] {\(\pi_1\)} (x'1);
			\node[rectangle callout,draw,inner sep=2pt,fill=colorp1,
				callout absolute pointer=(y1.east),
			below right= 25pt and 35pt of y1.north east]
			{\begin{tabular}{c c c c} \(e\) & \( = \)& \(e_1,\) & \(e_q\) \\ \rotatebox{-90}{\(\leqslant\)} & & \rotatebox{-90}{\(\leqslant\)}\(_{\pi_1}\) & \rotatebox{-90}{\(\leqslant\)}\(_{\pi_2}\) \\ \(e'\) & \( =\) & \(e'_1,\) & \(e'_q\)\end{tabular}};
		\end{tikzpicture}
	}
	\caption{Configurations Structures by the end of the proof of \autoref{prop:HHPB_congr}}
	\label{fig-HHPB_congr}
\end{figure}
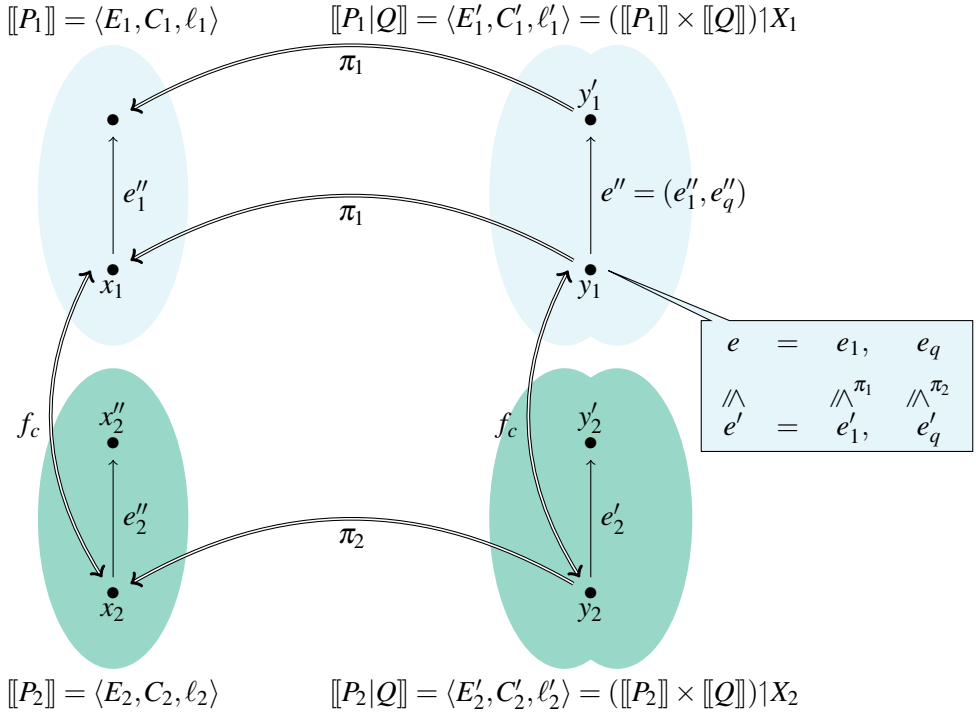

\begin{proof}
	We only consider the following case:
	\[
		\forall P_1, P_2, \enc{P_1}\sbisim\enc{P_2}\implies\forall Q, \enc{P_1\vert Q}\sbisim\enc{P_2\vert Q}
	\]
	
	As $\enc{P_1}\sbisim\enc{P_2}$, there exists $\srel$ a hereditary history preserving bisimulation (HHPB) between $\enc{P_1}$ and $\enc{P_2}$.
	\autoref{fig-HHPB_congr} introduces the variables names and types.
	
	Define $\rel_c\subseteq C'_1\times C'_2\times\power(E'_1\times E'_2)$ as follows:
	\begin{align*}
		(y_1, y_2, f_c) \in \rel_c \iff                                     
		\begin{dcases*}                                                       
		(\pi_1(y_1),\pi_2(y_2),\pi_1\circ f) \in \rel                       \\
		f_c(e)=(\pi_1 \circ f(e)),\pi_2(e))\in y_2\text{ for all }e \in y_1 
		\end{dcases*}                                                         
	\end{align*}
	
	Informally $(y_1, y_2, f_c)$ is in the relation $\rel_c$ if there is $(x_1, x_2, f)$ in $\rel$ such that $x_i$ is the first projection of $y_i$ and such that $f_c$ satisfies the property: for $(e_1, e_{q})\in E'_1$, $f_c(e_1,e_{q})=(f(e_1),e_{q})$ and $(f(e_1),e_{q})\in E'_2$.
	
	Let us show that $\rel_c$ is a HHPB between $\mem{E'_1, C'_1, \labl'_1}$ and $\mem{E'_2, C'_2, \labl'_2}$.
	\begin{itemize}
		\item $(\emptyset,\emptyset,\emptyset)\in\rel_c$;
		\item For $(y_1, y_2, f_c)\in\rel$ we show that $f_c$ is label and order preserving bijection. We have that $f_c$ is defined as $f_c(e)=(\pi_1 \circ f(e)),\pi_2(e))$, for some $f$ label and order preserving bijection such that $(\pi_1(y_1),\pi_2(y_2),\pi_1\circ f)\in\rel$.
		      
		      That $f_c$ is a bijection follows from $f$ a bijection.
		      
		      Let $e \in y_1$ with $\pi_1(e)=e_1$, $\pi_2(e)=e_{q}$, then $f_c(e)=(f(e_1),e_{q})$ for some $f_c$ \st $(\pi(y_1), \pi_2(y_2), f)\in\rel$.
		      We have that $\labl'_1(e)=(\labl_1(e_1),\labl_{Q}(e_{q}))$, hence
		      \[
		      	\labl'_2(f_c(e))=\labl'_2(f(e_1),e_{q})=\big(\labl_2(f(e_1)),\labl_{Q}(e_{q})\big)
		      \]
		      As $f$ is label preserving we get $\labl'_2(f_c(e))=(\labl_1(e_1),\labl_{Q}(e_{q}))$, hence $\labl'_1(e)=\labl'_2(f_c(e))$.
		      
		      Let us now show that for $e, e' \in y_1$, if $e\leq_{y_1} e'$ then $f_c(e)\leq_{y_2} f_c(e')$. We denote $\pi_1(e)=e_1$, $\pi_2(e)=e_{q}$ and $\pi_1(e')=e_1'$, $\pi_2(e')=e_{q}'$.
		      Then from \autoref{prop:cause_projection}
		      \begin{equation}
		      	e\leq_{y_1} e'\implies e_1\leq_{\pi_1(y_1)} e_1'\text{ or }e_{q}\leq_{\pi_2(y_1)} e_{q}'
		      \end{equation}
		      We consider the case where $e_1\leq_{\pi_1(y_1)} e_1'$. As $f$ is order preserving we have that $f(e_1)\leq_{\pi_1(y_2)}f(e_1')$. Then $(f(e_1),e_{q})\leq_{x_2}(f(e_1'),e_{q}')$, as the projections are order reflecting.
		      
		\item Let $(y_1, y_2,f_c)\in\rel_c$ and $y_1 \redl{e''} y_1'$, $y_1'=y_1\cup\set{e''}$. We consider only the case when $\pi_1(e'') = e''_1\neq\star$, $\pi_2(e'')= e''_{q}\neq\star$ as the rest is similar.
		      From the definition of the projections $\pi_1(y_1)$, $\pi_1(y_1')\in C_1'$ and as $\pi_1(e'')=e''_1\neq\star$, we have that $\pi_1(y_1')=\pi_1(y_1)\cup\set{e''_1}$.
		      We reason similarly on $\pi_2(y_1)$ and get
		      \begin{equation}
		      	\label{eq1}
		      	\pi_1(y_1)\redl{e''_1}\pi_1(y_1')\text{ and }\pi_2(y_1)\redl{e''_{q}}\pi_2(y_1').
		      \end{equation}
		      From \autoref{eq1} and as $(\pi_1(y_1), \pi_2(y_2),f)\in\rel$, by definition of \(\rel_c\), we have that
		      \begin{equation}
		      	\label{eq2}
		      	\exists x_2'\text{ \st }\pi_1(y_2)\redl{e''_2}x_2'=x_2\cup\set{e''_2}
		      \end{equation}
		      and
		      \begin{equation}\label{eq3}
		      	f'=f\cup\set{e_1''\leftrightarrow e_2''}
		      \end{equation}
		      such that $(x_1',x_2',f')\in\rel$.
		      From \autoref{eq1} and \autoref{eq2} we have that $\exists y'_2 \in (\enc{P_2}\times\enc{P_{Q}})$ with $\pi_1(y'_2)=x'_2$, $\pi_2(y'_2)=\pi_2(y_1')$ and $\exists e_2'\in y'_2$, $\pi_1(e'_2)=e''_2$, $\pi_2(e'_2)=e_{q}''$.
		      
		      Let us show that $y_2'\notin X_2$. We have that $y_2'\notin X_2$. As $\labl(e_1'')$ and $\labl(e_q'')$ are compatible, then so are $\labl(e_2'')$ and $\labl(e_q'')$, hence $y_2\cup\set{(e_2'',e_q'')}\notin X_2$.
		      
		      Remains to show $(y_1',y_2',f_c')\in\rel$, where $f_c'=f_c\cup\set{e''_1 \leftrightarrow e''_2}$. We have that $(\pi_1(y_1'),\pi_1(y_2'),f')\in\rel_c$ and from \autoref{eq3} that $\pi_1\circ f_c'=f'$. \qedhere
	\end{itemize}
\end{proof}
\end{document}